\documentclass[letterpaper, 10 pt, conference]{ieeeconf}
\IEEEoverridecommandlockouts

\usepackage{geometry}
\newgeometry{top=0.75in, bottom=0.75in, right=0.75in, left=0.75in}

\usepackage{amsmath,amssymb,amsfonts}

\usepackage{amsthm}
\usepackage{mathtools}
\usepackage{bm}

\usepackage{graphicx}
\usepackage{subcaption}
\usepackage{booktabs}
\usepackage{multirow}

\usepackage{algorithm}
\usepackage{algorithmic}

\usepackage[hidelinks]{hyperref}
\hypersetup{
  pdftitle={Data-Driven Reachability Analysis via Diffusion Models},
  pdfauthor={Yanliang Huang, Peng Xie, Wenyuan Wu, Zhuoqi Zeng, Amr Alanwar},
  pdfsubject={Data-driven reachability analysis, diffusion models, conformal prediction},
  pdfkeywords={Reachability analysis, diffusion models, conformal prediction, PAC guarantee}
}
\usepackage{url}
\usepackage{cite}

\usepackage{xcolor}
\usepackage{tikz}

\newtheoremstyle{cdcplain}{}{}{}{}{\itshape}{.}{ }{}
\newtheoremstyle{cdcdefn}{}{}{}{}{\itshape}{.}{ }{}

\theoremstyle{cdcplain}
\newtheorem{theorem}{Theorem}

\newtheorem{proposition}{Proposition}

\theoremstyle{cdcdefn}
\newtheorem{definition}{Definition}
\newtheorem{remark}{Remark}
\newtheorem{assumption}{Assumption}
\newtheorem{problem}{Problem}

\newcommand{\R}{\mathbb{R}}
\newcommand{\E}{\mathbb{E}}
\newcommand{\Prob}{\mathbb{P}}

\renewcommand{\QED}{\QEDopen}

\graphicspath{{figures/}}

\title{\LARGE \bf
Data-Driven Reachability Analysis via Diffusion Models \\ with PAC Guarantees
}

\author{Yanliang Huang$^{1}$, Peng Xie$^{1}$, Wenyuan Wu$^{1}$, Zhuoqi Zeng$^{2}$, and Amr Alanwar$^{1}$%
\thanks{$^{1}$School of Computation, Information and Technology,
        Technical University of Munich, Munich, Germany.
        \texttt{\{yanliang.huang, p.xie, wenyuan.wu,
        alanwar\}@tum.de}}%
\thanks{$^{2}$School of Engineering,
        Hainan Bielefeld University of Applied Sciences, Hainan, China.
        \texttt{zhuoqi.zeng@hainan-biuh.edu.cn}}%
}

\begin{document}

\maketitle

\begin{abstract}
We present a data-driven framework for reachability analysis
of nonlinear dynamical systems that requires no explicit model.
A denoising diffusion probabilistic model learns the
time-evolving state distribution of a dynamical system
from trajectory data alone.
The predicted reachable set takes the form of a sublevel set
of a nonconformity score derived from the reconstruction error,
with the threshold calibrated via the Learn Then Test procedure
so that the probability of excluding a reachable state
is bounded with high probability.
Experiments on three nonlinear systems, a forced Duffing oscillator,
a planar quadrotor, and a high-dimensional reaction-diffusion system,
confirm that the empirical miss rate remains below
the Probably Approximately Correct (PAC) bound while scaling to state dimensions
beyond the reach of classical grid-based and
polynomial methods.
\end{abstract}

\section{Introduction}
\label{sec:introduction}

\begin{figure*}[t]
    \centering
    \includegraphics[width=\textwidth]{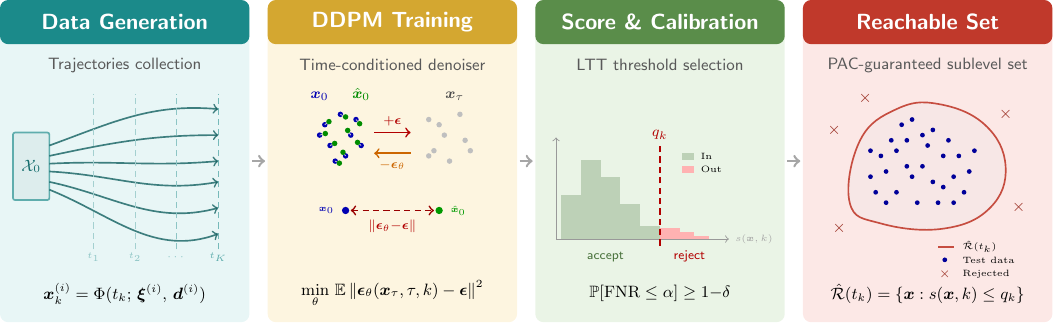}
    \caption{Pipeline overview for the data-driven reachability analysis via diffusion models. From left to right: (1)~trajectory data is collected from the nonlinear dynamical system at discrete time steps; (2)~a time-conditioned DDPM is trained to denoise states, with the reconstruction error serving as a nonconformity score; (3)~the LTT procedure selects a threshold~$q_k$ for the nonconformity score that controls the false negative rate with PAC guarantee; (4)~the predicted reachable set is the sublevel set of all states whose score does not exceed~$q_k$.}
    \label{fig:pipeline}
\end{figure*}

Reachability analysis determines the set of all states
a dynamical system can visit from a given set of initial conditions \cite{althoff2021set}.
The resulting reachable sets are central to safety
verification and control design.
Classical set-propagation methods compute
tight over-approximations of the reachable set, yet they
rely on an explicit analytical model of the system dynamics,
which may be unavailable or prohibitively complex
for high-dimensional or black-box systems.

Several data-driven methods avoid this dependence on
analytical models, but existing approaches with
finite-sample guarantees either rely on learned surrogate
models whose coverage degrades with model error, or
employ polynomial representations that scale poorly
with state dimension.
A method that bypasses surrogate construction,
provides finite-sample coverage guarantees,
and scales to high-dimensional states
has, to our knowledge, not been demonstrated.

\subsection{Related Work}
\label{sec:related_work}

Among model-based approaches, set-propagation tools
represent reachable sets with
zonotopes~\cite{girard2005reachability},
polynomial zonotopes~\cite{althoff2015cora},
or constrained polynomial
zonotopes~\cite{zhang2025data}.
Hamilton-Jacobi PDE methods~\cite{chen2018hamilton}
provide exact characterizations but scale poorly beyond
moderate dimensions.
Allen et al.~\cite{shia2014svm} use support vector
machines trained on pre-computed reachable/unreachable
labelings for real-time classification, but still require an
analytical model to generate the labels.
DeepReach~\cite{bansal2021deepreach} parameterizes the HJ
value function with a neural network, reducing memory
cost but without finite-sample coverage guarantees.

Several data-driven works compute statistical over-approximations
using matrix zonotopes~\cite{alanwar2023data},
Gaussian processes~\cite{devonport2020data},
or scenario-optimized sublevel
sets~\cite{dietrich2024scenario}, each with
probabilistic or deterministic containment guarantees.
Devonport et al.~\cite{devonport2021christoffel}
formalize probabilistic reachability and construct
polynomial sublevel sets via Christoffel functions
with VC-dimension sample bounds.
Devonport et al.~\cite{devonport2022christoffel_pac}
extend this with PAC-Bayes bounds.  Both remain
limited by a monomial basis that grows
combinatorially with state dimension.

Several recent works combine data-driven
reachability with conformal prediction.
Tebjou et al.~\cite{tebjou2023christoffel} augment
Christoffel-function sublevel sets with conformal
calibration.
Surrogate-based approaches learn a dynamics model
and apply conformal inference to quantify prediction
error~\cite{hashemi2023data,hashemi2024statistical,nath2026koopman}.
Huang et al.~\cite{huang2026cddr} combine
Learn Then Test (LTT)~\cite{angelopoulos2022learn}
with zonotope-based residual scores.
In contrast, our framework learns the state distribution
directly with a generative model, avoiding surrogate
construction entirely.

This paper fills this gap by combining two ideas.
The reconstruction error of
DDPMs~\cite{ho2020denoising,sohldickstein2015deep}
serves as an anomaly score~\cite{wyatt2022anoddpm}
that approximates the negative log-likelihood via the
evidence lower bound~\cite{li2023diffusion}.
The LTT framework~\cite{angelopoulos2022learn}
calibrates this score into PAC prediction sets,
yielding tighter thresholds than Bonferroni-corrected
quantile methods~\cite{vovk2005algorithmic,bates2021testing,sadinle2019least}.
The predicted reachable set is therefore the sublevel set
of the calibrated score
(Fig.~\ref{fig:pipeline}).

\subsection{Contribution}
The contributions are as follows:
\begin{enumerate}
  \item A data-driven reachability analysis framework
    with PAC coverage guarantees that scales to
    state dimensions intractable for
    existing grid-based and polynomial methods.
    A density lower-bound assumption further yields
    a geometric volume bound on missed states.
  \item Systematic experimental comparison on
    a forced Duffing oscillator, a planar quadrotor, 
    and a Gray-Scott reaction-diffusion system
    with 8192-dimensional states.
    The diffusion-based score
    yields the tightest predicted sets among all
    tested score functions.
\end{enumerate}

\section{Preliminaries and Problem Statement}
\label{sec:problem}

\subsection{Dynamical Systems and Reachable Sets}

Consider a dynamical system with a state transition function
$\Phi(t;\, \bm{\xi},\, \bm{d})$ that maps an initial state
$\bm{\xi} \in \R^n$ to the state at time~$t$
under an external input
$\bm{d} \colon [0,T_f] \to \R^m$.

Given an initial set $\mathcal{X}_0 \subset \R^n$
and a set of admissible external inputs~$\mathcal{W}$,
the forward reachable set is defined as follows.

\begin{definition}[Forward Reachable Set]
    \label{def:reach_set}
    The forward reachable set at time~$t$ is
    \begin{equation}
      \mathcal{R}(t)
      = \bigl\{\Phi(t;\, \bm{\xi},\, \bm{d})
        : \bm{\xi} \in \mathcal{X}_0,\;
          \bm{d} \in \mathcal{W} \bigr\}.
      \label{eq:reach_set}
    \end{equation}
\end{definition}

To formulate a probabilistic version of the reachability
problem, we introduce random variables
$\bm{\Xi}$ and $\bm{D}$ taking values in
$\mathcal{X}_0$ and $\mathcal{W}$,
respectively~\cite{devonport2021christoffel}.
At each discrete time step
$k \in \{0,\ldots,K{-}1\}$ with $t_k = k\,\Delta t$,
the composed random variable
$\Phi(t_k;\, \bm{\Xi},\, \bm{D})$ induces a probability
measure~$\mu_k$ on~$\R^n$.
When the distributions of $\bm{\Xi}$ and $\bm{D}$
have full support on $\mathcal{X}_0$ and $\mathcal{W}$,
the support of $\mu_k$ coincides with $\mathcal{R}(t_k)$.

The analytical form of~$\Phi$ is assumed unavailable.
A dataset of~$N$ independent trajectories is generated
by independently sampling an initial condition
$\bm{\xi}^{(i)} \sim p_0$ with
$\mathrm{supp}(p_0) \subseteq \mathcal{X}_0$
and an external input
$\bm{d}^{(i)}(\cdot) \sim p_{\bm{d}}$ with
$\mathrm{supp}(p_{\bm{d}}) \subseteq \mathcal{W}$,
and evaluating
$\bm{x}_k^{(i)} = \Phi(t_k;\,\bm{\xi}^{(i)},\,
\bm{d}^{(i)})$ for $k = 0, \ldots, K{-}1$,
following~\cite{devonport2021christoffel}.

\subsection{Problem Statement}

\begin{problem}
    \label{prob:pac_reach}
    Given the state transition function~$\Phi$,
    distributions~$p_0$ on~$\mathcal{X}_0$ and
    $p_{\bm{d}}$ on~$\mathcal{W}$,
    a risk level $\alpha \in (0,1)$, and a confidence
    parameter $\delta \in (0,1)$, construct predicted
    reachable sets
    $\hat{\mathcal{R}}(t_k) \subseteq \R^n$ for
    $k = 0, \ldots, K{-}1$ such that
    \begin{equation}
      \Prob_{\mathrm{cal}}\!\Bigl[
        \forall\, k \in \{0,\ldots,K{-}1\}:
        \Prob_{\bm{x} \sim \mu_k}\!\bigl[
          \bm{x} \notin \hat{\mathcal{R}}(t_k)
        \bigr]
        \leq \alpha
      \Bigr]
      \geq 1 - \delta,
      \label{eq:pac_guarantee}
    \end{equation}
    where the outer probability is over the random
    dataset used to construct~$\hat{\mathcal{R}}(t_k)$
    and $\bm{x} \sim \mu_k$ is drawn independently
    of that dataset.
\end{problem}

\section{Method}
\label{sec:method}

The framework has two stages:
learning the time-conditioned state distribution with a
diffusion model and calibrating a rejection threshold via
LTT to construct the predicted
reachable set.

\subsection{Reachability Analysis via Diffusion Models}
\label{sec:method_calibration}

The generative backbone is a
DDPM~\cite{ho2020denoising, sohldickstein2015deep}.
Its forward process corrupts a clean state $\bm{x}_0$
by adding Gaussian noise over $T$ diffusion steps
with variance schedule
$\{\beta_\tau\}_{\tau=1}^{T}$:
\begin{equation}
    q(\bm{x}_\tau \mid \bm{x}_0)
    = \mathcal{N}\!\bigl(\bm{x}_\tau;\,
    \sqrt{\bar\alpha_\tau}\,\bm{x}_0,\;
    (1{-}\bar\alpha_\tau)\,\bm{I}\bigr),
    \label{eq:forward_diffusion}
\end{equation}
where $\alpha_\tau = 1 - \beta_\tau$ and
$\bar\alpha_\tau = \prod_{s=1}^{\tau}\alpha_s$.
Equivalently,
\begin{equation}
    \bm{x}_\tau
    = \sqrt{\bar\alpha_\tau}\,\bm{x}_0
    + \sqrt{1{-}\bar\alpha_\tau}\,\bm{\epsilon},
    \quad
    \bm{\epsilon} \sim \mathcal{N}(\bm{0},\bm{I}).
    \label{eq:noisy_sample}
\end{equation}
A neural network $\bm{\epsilon}_\theta(\bm{x}_\tau,\tau,k)$
is trained to predict the added noise from the corrupted
sample, conditioned on both the diffusion timestep~$\tau$
and the physical time index~$k$
(Algorithm~\ref{alg:pipeline}, line~6).

After training, the model implicitly encodes the support
of each distribution~$\mu_k$.
A state inside $\mathcal{R}(t_k)$ produces a low
reconstruction error, while an out-of-distribution state
gives a large one.
The nonconformity score of a pair $(\bm{x},k)$ is
\begin{equation}
    s(\bm{x},k)
    = \sum_{\tau=1}^{T} \gamma_\tau\,
    \bigl\lVert
    \bm{\epsilon}_\theta(\bm{x}_{\tau},\tau,k)
    - \bm{\epsilon}_\tau
    \bigr\rVert^2,
    \label{eq:score}
\end{equation}
where $\gamma_\tau$ are non-negative weights
and $\bm{\epsilon}_\tau \sim \mathcal{N}(\bm{0},\bm{I})$.
Proposition~\ref{prop:score_density} shows that
the ELBO-weighted score equals the negative
ELBO plus a state-dependent quadratic term
and a constant.
Practical weight choices are discussed in
Section~\ref{sec:impl_details}.
Each $\bm{x}_{\tau}$ is constructed
via~\eqref{eq:noisy_sample} using
the same $\bm{\epsilon}_\tau$ that appears as the
prediction target in~\eqref{eq:score}.

\begin{proposition}[Score--likelihood correspondence]
    \label{prop:score_density}
    Let $\bm{\epsilon}^*$ be the Bayes-optimal denoiser
    and assume Gaussian parameterization at all
    diffusion steps.
    Define the ideal score
    $s^*(\bm{x},k) = \sum_{\tau=1}^{T} \gamma_\tau
    \lVert\bm{\epsilon}^*(\bm{x}_\tau,\tau,k)
    - \bm{\epsilon}_\tau\rVert^2$,
    where $\gamma_\tau$ are the loss weights in the
    ELBO decomposition of~\cite{ho2020denoising}.
    Then its conditional expectation over the noise sequence
    $\bm{\epsilon}_{1:T} = (\bm{\epsilon}_1, \ldots, \bm{\epsilon}_T)$
    satisfies
    \begin{equation}
        \E_{\bm{\epsilon}_{1:T}}\bigl[
        s^*(\bm{x},k)\mid \bm{x},k\bigr]
        = -\mathrm{ELBO}(\bm{x}\mid k)
        - \tfrac{1}{2}\bar\alpha_T
        \lVert\bm{x}\rVert^2
        + C,
        \label{eq:score_density}
    \end{equation}
    where $C$ is a constant independent of~$\bm{x}$
and $\bar\alpha_T = \prod_{s=1}^{T}\alpha_s$.
\end{proposition}

\begin{proof}
    The ELBO decomposition of
    Ho et al.~\cite{ho2020denoising} yields
    \begin{align*}
        -\mathrm{ELBO}(\bm{x}\mid k)
        &= D_{\mathrm{KL}}\!\bigl(
        q(\bm{x}_T\mid\bm{x})\,\|\, p(\bm{x}_T)\bigr)
        \\
        &\quad + \sum_{\tau=1}^{T} \gamma_\tau\,
        \E_{\bm{\epsilon}_\tau}\!\bigl[
        \lVert\bm{\epsilon}^*(\bm{x}_\tau,\tau,k)
        - \bm{\epsilon}_\tau\rVert^2
        \mid \bm{x},k\bigr]\\
        &\quad + C',
    \end{align*}
    where each constant~$C_\tau$ depends only on the
    noise schedule $\beta_\tau, \bar\alpha_\tau$, so
    $C' = \sum_\tau C_\tau$ is independent of~$\bm{x}$.
    Identifying the weighted sum with
    $\E_{\bm{\epsilon}_{1:T}}[s^*(\bm{x},k)\mid \bm{x},k]$,
    it remains to evaluate the prior term.
    Since
    $q(\bm{x}_T\mid\bm{x})
    = \mathcal{N}(\sqrt{\bar\alpha_T}\,\bm{x},\,
    (1{-}\bar\alpha_T)\bm{I})$
    and $p(\bm{x}_T) = \mathcal{N}(\bm{0},\bm{I})$,
    the Gaussian KL formula gives
    \[
        D_{\mathrm{KL}}\!\bigl(
        q(\bm{x}_T\mid\bm{x})\,\|\, p(\bm{x}_T)\bigr)
        = \tfrac{1}{2}\bigl[
        \bar\alpha_T\lVert\bm{x}\rVert^2
        - n\bar\alpha_T
        - n\ln(1{-}\bar\alpha_T)\bigr],
    \]
    where $n = \dim(\bm{x})$ is the state dimension.
    Notice that $-n\bar\alpha_T - n\ln(1{-}\bar\alpha_T)$
    depends only on the state dimension~$n$ and the noise schedule,
    thus forming a constant $C_{\mathrm{prior}}$ independent of~$\bm{x}$.
    Therefore,
    \[
        D_{\mathrm{KL}}\!\bigl(
        q(\bm{x}_T\mid\bm{x})\,\|\, p(\bm{x}_T)\bigr)
        = \tfrac{1}{2}\bar\alpha_T
        \lVert\bm{x}\rVert^2 + C_{\mathrm{prior}}.
    \]
    Substituting this into the ELBO decomposition and
    rearranging for $\E_{\bm{\epsilon}_{1:T}}[s^*(\bm{x},k)\mid \bm{x},k]$
    yields~\eqref{eq:score_density}, where all constant
    terms are absorbed into~$C$.
\end{proof}

\begin{remark}\label{rem:score_interpretation}
    Since $\mathrm{ELBO}(\bm{x}\mid k)
    \leq \log p(\bm{x}\mid k)$,
    the score sublevel sets approximate the density
    level sets of~$\mu_k$, which are
    Neyman--Pearson optimal.  For typical noise schedules
    with $\bar\alpha_T \approx 0$, the additive terms
    in~\eqref{eq:score_density} approximately preserve
    score ranking and thus do not affect conformal
    calibration.
\end{remark}

\subsection{Conformal Prediction via LTT}
\label{sec:conformal_ltt}

The score~\eqref{eq:score} by itself does not provide a
coverage guarantee. A calibrated threshold is needed to
convert the continuous score into a binary accept/reject
decision with controlled error rate.
The statistical guarantee requires that calibration
samples at each step~$k$ are i.i.d.\ draws
from~$\mu_k$, which holds when trajectories are
generated from independent initial conditions
and external inputs.
Time-dependent thresholds
$\{q_k\}_{k=0}^{K-1}$ are obtained through the
LTT framework~\cite{angelopoulos2022learn}.
For each candidate threshold~$\lambda_l$ from a finite
grid~$\Lambda_k$, the empirical miss rate counts the
fraction of calibration scores exceeding~$\lambda_l$
(Algorithm~\ref{alg:pipeline}, line~13).
A Hoeffding--Bentkus p-value
(Algorithm~\ref{alg:pipeline}, line~14) combines a
sub-Gaussian concentration bound with the binomial CDF,
yielding a tighter test than either
alone~\cite[Thm.~1]{angelopoulos2022learn}.
Because the empirical risk is non-increasing
in~$\lambda_l$, no additional multiplicity correction
beyond LTT is needed.
The selected threshold and predicted reachable set are
\begin{equation}
    q_k = \min\bigl\{\lambda_l \in \Lambda_k
    : p_k(\lambda_l) \leq \delta / K\bigr\},
    \label{eq:threshold}
\end{equation}
\begin{equation}
    \hat{\mathcal{R}}(t_k)
    = \bigl\{\bm{x} \in \R^n : s(\bm{x},k) \leq q_k\bigr\}.
    \label{eq:reach_pred}
\end{equation}

\begin{algorithm}[t]
\caption{Data-driven reachability analysis
via diffusion models}
\label{alg:pipeline}
\begin{algorithmic}[1]
\renewcommand{\algorithmicrequire}{\textbf{Input:}}
\renewcommand{\algorithmicensure}{\textbf{Output:}}
\REQUIRE
State transition function~$\Phi$;
distributions~$p_0$ on~$\mathcal{X}_0$
and~$p_{\bm{d}}$ on~$\mathcal{W}$;
trajectory count~$N$;
risk level~$\alpha$;
confidence~$\delta$;
score weights~$\{\gamma_\tau\}$;
threshold grid~$\Lambda$
\ENSURE Predicted reachable sets
$\{\hat{\mathcal{R}}(t_k)\}_{k=0}^{K-1}$
\FORALL{$i \in \{1, \ldots, N\}$}
    \STATE Sample
    $\bm{\xi}^{(i)} \sim p_0$,\;
    $\bm{d}^{(i)} \sim p_{\bm{d}}$
    \STATE $\bm{x}_k^{(i)} \gets
    \Phi(t_k;\,\bm{\xi}^{(i)},\,\bm{d}^{(i)})$,\;
    $k = 0, \ldots, K{-}1$
\ENDFOR
\STATE Split
$\{\bm{x}_k^{(i)}\}_{i,k}$ into
$\mathcal{D}_{\mathrm{tr}}$,
$\mathcal{D}_{\mathrm{cal}}$
\STATE $\bm{\epsilon}_\theta \gets
\arg\min_\theta\;
\E_{\bm{x}_0,\bm{\epsilon},\tau,k}
\lVert\bm{\epsilon}_\theta(\bm{x}_\tau,\tau,k)
- \bm{\epsilon}\rVert^2$
on $\mathcal{D}_{\mathrm{tr}}$
\FOR{$k = 0,\ldots,K{-}1$}
    \FOR{$i = 1,\ldots,n_k$ in
    $\mathcal{D}_{\mathrm{cal}}$}
        \STATE Sample
        $\{\bm{\epsilon}_\tau\}_{\tau=1}^{T}
        \sim \mathcal{N}(\bm{0},\bm{I})$
        \STATE $\bm{x}^{(i)}_{\tau} \gets
        \sqrt{\bar\alpha_{\tau}}\,\bm{x}^{(i)}
        + \sqrt{1{-}\bar\alpha_{\tau}}\,
        \bm{\epsilon}_\tau$,\;
        $\tau{=}1,\ldots,T$
        \STATE $s_i \gets
        \sum_{\tau=1}^{T} \gamma_\tau
        \lVert\bm{\epsilon}_\theta(
        \bm{x}^{(i)}_{\tau},\tau,k)
        - \bm{\epsilon}_\tau\rVert^2$
    \ENDFOR
    \STATE $\hat{r}_k(\lambda) \gets
    \tfrac{1}{n_k}\sum_{i=1}^{n_k}
    \mathbf{1}[s_i {>} \lambda]$,\;
    $\forall\lambda \in \Lambda$
    \STATE $p_k(\lambda) \gets$
    HB p-value of $\hat{r}_k(\lambda)$
    \STATE $q_k \gets \min\{\lambda \in \Lambda :
    p_k(\lambda) \leq \delta/K\}$
    \STATE $\hat{\mathcal{R}}(t_k) \gets
    \{\bm{x}\in\R^n : s(\bm{x},k) \leq q_k\}$
\ENDFOR
\RETURN $\{\hat{\mathcal{R}}(t_k)\}_{k=0}^{K-1}$
\end{algorithmic}
\end{algorithm}

\begin{theorem}[PAC Reachability Coverage]
    \label{thm:pac}
    Assume the $n_k$ calibration samples at each step~$k$
    are i.i.d.\ draws from~$\mu_k$, and let the
    thresholds $\{q_k\}_{k=0}^{K-1}$ be computed
    via~\eqref{eq:threshold} with confidence
    budget~$\delta \in (0,1)$ and risk level
    $\alpha \in (0,1)$.
    Then the predicted reachable
    set~\eqref{eq:reach_pred} satisfies
    \begin{equation}
        \Prob_{\mathrm{cal}}\!\Bigl[
        \forall\, k \in \{0,\ldots,K{-}1\}:
        \Prob_{\bm{x} \sim \mu_k}\!
        \bigl[\bm{x} \notin \hat{\mathcal{R}}(t_k)\bigr]
        \leq \alpha
        \Bigr]
        \geq 1 - \delta.
        \label{eq:pac_theorem}
    \end{equation}
\end{theorem}

\begin{proof}
    For each time step~$k$, the LTT
    procedure~\cite{angelopoulos2022learn} with
    budget~$\delta_k = \delta/K$ guarantees
    $\Prob_{\mathrm{cal}}\bigl[
    \Prob_{\bm{x}\sim\mu_k}(\bm{x}\notin
    \hat{\mathcal{R}}(t_k)) > \alpha\bigr]
    \leq \delta/K$.
    A union bound over the $K$ time steps yields
    \[
        \Prob_{\mathrm{cal}}\!\Bigl[
        \exists\, k:
        \Prob_{\bm{x}\sim\mu_k}\!
        \bigl[\bm{x}\notin\hat{\mathcal{R}}(t_k)\bigr]
        > \alpha\Bigr]
        \leq \sum_{k=0}^{K-1}\frac{\delta}{K}
        = \delta,
    \]
    and complementing gives~\eqref{eq:pac_theorem}.
\end{proof}

Tighter joint corrections such as Holm--Bonferroni
could reduce conservatism but would require ordering
the p-values across time steps.

\subsection{From Distributional to Geometric Coverage}
\label{sec:geometric}

\begin{assumption}[Positive initial density]
    \label{assum:density}
    The initial distribution~$p_0$ satisfies
    $p_0(\bm{\xi}) \geq c_0 > 0$ for all
    $\bm{\xi} \in \mathcal{X}_0$.
    In particular,
    $\mathrm{supp}(p_0) = \mathcal{X}_0$.
\end{assumption}

\begin{proposition}[Set-theoretic volume bound]
    \label{prop:volume}
    Let Assumption~\ref{assum:density} hold,
    assume the calibration samples are i.i.d.\
    from~$\mu_k$, and suppose that,
    for each input~$\bm{d} \in \mathcal{W}$
    and each time step~$k$, the map
    $\bm{\xi} \mapsto \Phi(t_k;\,\bm{\xi},\,\bm{d})$
    is a diffeomorphism.
    Define
    $J_k = \sup_{\bm{\xi} \in \mathcal{X}_0,\,
    \bm{d} \in \mathcal{W}}
    \lvert\det(D_{\bm{\xi}}
    \Phi(t_k;\,\bm{\xi},\,\bm{d}))\rvert$.
    Then, under the event in
    Theorem~\ref{thm:pac},
    \begin{equation}
        \E_{\bm{D} \sim p_{\bm{d}}}\!\bigl[
        \mathrm{vol}\bigl(
        \Phi(t_k;\,\mathcal{X}_0,\,\bm{D})
        \setminus \hat{\mathcal{R}}(t_k)
        \bigr)\bigr]
        \leq \frac{\alpha\, J_k}{c_0},
        \quad \forall\, k.
        \label{eq:volume_bound}
    \end{equation}
\end{proposition}

\begin{proof}
    Fix~$k$ and let
    $A_k = \mathcal{R}(t_k) \setminus
        \hat{\mathcal{R}}(t_k)$.
    For each $\bm{d} \in \mathcal{W}$, write
    $\Phi_k^{\bm{d}}(\cdot)
    = \Phi(t_k;\, \cdot,\, \bm{d})$.
    Because $\Phi_k^{\bm{d}}$ is a diffeomorphism,
    the change-of-variables formula gives the
    conditional pushforward density
    \[
        p_{k \mid \bm{d}}(\bm{x})
        = \frac{p_0\!\bigl(
            (\Phi_k^{\bm{d}})^{-1}(\bm{x})\bigr)}
               {\bigl\lvert\det\bigl(
                   D\Phi_k^{\bm{d}}\!\bigl(
                   (\Phi_k^{\bm{d}})^{-1}(\bm{x})\bigr)
               \bigr)\bigr\rvert}
        \geq \frac{c_0}{J_k},
    \]
    for every
    $\bm{x} \in \Phi(t_k;\,\mathcal{X}_0,\,\bm{d})$,
    by Assumption~\ref{assum:density} and the
    definition of~$J_k$.
    The marginal density is
    $p_k(\bm{x})
    = \int_{\mathcal{W}} p_{k \mid \bm{d}}(\bm{x})\,
    p_{\bm{d}}(\bm{d})\, d\bm{d}$.
    Theorem~\ref{thm:pac} ensures that, with
    probability $\geq\!1{-}\delta$ over the
    calibration data, simultaneously for
    all~$k$, $\mu_k(A_k) \leq \alpha$.
    Under this event, applying
    Fubini's theorem and the density lower bound,
    \begin{align*}
        \alpha &\geq \mu_k(A_k)
        = \int_{A_k}\! p_k(\bm{x})\,d\bm{x}\\
        &= \int_{\mathcal{W}}
           \!\int_{A_k \cap
           \Phi(t_k;\,\mathcal{X}_0,\,\bm{d})}
           \!p_{k \mid \bm{d}}(\bm{x})\,
           d\bm{x}\; p_{\bm{d}}(\bm{d})\,d\bm{d}\\
        &\geq \frac{c_0}{J_k}
           \int_{\mathcal{W}}
           \mathrm{vol}\bigl(
           \Phi(t_k;\,\mathcal{X}_0,\,\bm{d})
           \setminus \hat{\mathcal{R}}(t_k)
           \bigr)\,p_{\bm{d}}(\bm{d})\,d\bm{d},
    \end{align*}
    and rearranging yields~\eqref{eq:volume_bound}.
\end{proof}

\begin{remark}[Geometric coverage for dissipative systems]
    \label{rem:volume}
    For systems whose ODE vector field~$f$ has
    constant divergence
    $\mathrm{div}(f) \equiv -c \leq 0$,
    Liouville's formula
    \[
        \det(D_{\bm{\xi}}\Phi(t;\,\bm{\xi},\,\bm{d}))
        = \exp\!\biggl(\int_0^{t}
        \mathrm{div}(f)(\Phi(s;\,\bm{\xi},\,\bm{d}))\,
        ds\biggr)
    \]
    gives $J_k = e^{-c\,t_k} \leq 1$,
    so the bound~\eqref{eq:volume_bound}
    tightens over time.
\end{remark}

\section{Experiments}
\label{sec:experiments}

\subsection{Experimental Setup}
\label{sec:impl_details}

We evaluate the framework on three systems of increasing
dimensionality:
a forced Duffing oscillator ($n{=}2$),
a planar quadrotor ($n{=}6$),
and the Gray-Scott reaction-diffusion PDE ($n{=}8{,}192$).

All experiments generate trajectories from initial
conditions drawn independently from $p_0$ supported
on~$\mathcal{X}_0$.
The $N$ trajectories are randomly partitioned into
60\% training, 20\% calibration, and 20\% test
splits at the trajectory level, so that calibration
states at each step~$k$ are i.i.d.\ draws
from~$\mu_k$.
All experiments share the same PAC parameters:
risk level $\alpha{=}0.10\%$ and confidence
parameter $\delta{=}0.2$.

For calibration,
the threshold grid~$\Lambda_k$ at each step~$k$ is a
uniform partition of the interval $[s_{\min,k}, s_{\max,k}]$
into $L{=}2{,}000$ equally spaced candidate thresholds,
where $s_{\min,k}$ and $s_{\max,k}$ are the minimum and
maximum calibration scores at step~$k$.
Each grid is constructed before the LTT procedure begins
and is recomputed identically for every experiment.

To quantify prediction tightness we report
Intersection over Union (IoU) and precision
against a reference reachable set constructed
from test trajectories, where a grid point is
labeled reachable if any test trajectory visits
its Voronoi cell:
\begin{align}
    \mathrm{IoU} &= \frac{
        \lvert\hat{\mathcal{R}}(t_k) \cap
        \mathcal{R}_{\mathrm{ref}}(t_k)\rvert
    }{
        \lvert\hat{\mathcal{R}}(t_k) \cup
        \mathcal{R}_{\mathrm{ref}}(t_k)\rvert
    },
    \label{eq:iou} \\
    \mathrm{Precision} &= \frac{
        \lvert\hat{\mathcal{R}}(t_k) \cap
        \mathcal{R}_{\mathrm{ref}}(t_k)\rvert
    }{
        \lvert\hat{\mathcal{R}}(t_k)\rvert
    },
    \label{eq:precision}
\end{align}
where $\mathcal{R}_{\mathrm{ref}}(t_k)$ is the
reference set and $\lvert\cdot\rvert$ denotes
the number of grid cells.

\subsubsection{Score Design Choices}
The score~\eqref{eq:score} is restricted to
a small set of low-noise diffusion steps,
which both improves discrimination and
reduces the per-query cost to only a few
forward passes through the denoising network.
Substituting the forward process
$\bm{x}_\tau = \sqrt{\bar\alpha_\tau}\bm{x}_0
+ \sqrt{1{-}\bar\alpha_\tau}\bm{\epsilon}$
into the Bayes-optimal prediction
$\bm{\epsilon}^*
= (\bm{x}_\tau - \sqrt{\bar\alpha_\tau}
\hat{\bm{x}}_0)/\sqrt{1{-}\bar\alpha_\tau}$
shows that the per-step error decomposes as
\begin{equation}
    h_\tau(\bm{x},k)
    = \lVert \bm{\epsilon}^* {-} \bm{\epsilon}
    \rVert^2
    = \mathrm{SNR}(\tau)
    \lVert \bm{x} - \hat{\bm{x}}_0 \rVert^2,
    \label{eq:snr_decomp}
\end{equation}
where $\mathrm{SNR}(\tau)
= \bar\alpha_\tau/(1{-}\bar\alpha_\tau)$
and $\hat{\bm{x}}_0
= \E[\bm{x}_0 \mid \bm{x}_\tau, k]$.
At high noise $\mathrm{SNR}(\tau) \to 0$,
which collapses the posterior mean to
$\E_{\mu_k}[\bm{x}_0]$ and renders the score
uninformative. Finite model capacity and
per-sample noise variance amplify this
degradation but are negligible at low noise.

We therefore select only three low-noise diffusion
timesteps $\tau \in \{1, 2, 3\}$ out of $T{=}1{,}000$.
Each timestep is evaluated with an independent noise
realization and the results are averaged,
reducing score variance due to the random
$\bm{\epsilon}$.
Each timestep is repeated $R{=}8$ times, giving
$M{=}3 \times 8 = 24$ forward passes per query.
Both the timestep range and the number of
repeats are selected through ablation.
We weight all selected timesteps equally,
$\gamma_\tau = 1/M$.
ELBO weighting
$\gamma_\tau \propto \beta_\tau/(1{-}\bar\alpha_\tau)$
reduces IoU to $0.605$
because the finite-capacity model approximates the
Bayes-optimal denoiser poorly when the SNR is low,
the training objective uses uniform rather than
ELBO weights, and a single noise realization
introduces variance that grows as the SNR decreases.

\subsubsection{Baseline Selection}
Since the PAC guarantee of Theorem~\ref{thm:pac}
holds for any nonconformity score, we instantiate
the framework with multiple score functions to
isolate the effect of score design on tightness.
Three baselines are considered:
(i)~a normalizing flow (NF),
(ii)~a variational autoencoder (VAE),
and (iii)~the Christoffel-function
method~\cite{devonport2021christoffel}.
NF is compared on the 2D Duffing system,
Christoffel on the Duffing and quadrotor systems,
and VAE on the 8192-dimensional Gray-Scott system.
All baselines except Christoffel are calibrated via LTT.

\subsection{Duffing Oscillator}
\label{sec:duffing}

The forced Duffing oscillator has the state
$\bm{x} = (x,\dot{x})^\top \in \R^2$ and dynamics
\begin{equation}
    \ddot{x} + c\dot{x} - a x + b x^3
    = A\cos(\omega t),
    \label{eq:duffing}
\end{equation}
with parameters
$a{=}1$, $b{=}5$, $c{=}0.02$,
$A{=}8$, $\omega{=}0.5$.
The initial set is
$\mathcal{X}_0 = [-1,1]^2$.
We generate $N{=}10^6$ trajectories recorded at
$K{=}300$ time steps ($\Delta t{=}0.1$).

The denoising network is a 12-layer FiLM-MLP
with $T{=}1{,}000$ diffusion steps,
linear $\beta$ schedule, hidden
dimension 2048 and 128-dimensional time embeddings,
trained for 100 epochs with
AdamW, lr${=}5{\times}10^{-4}$, batch size $131{,}072$.
We compare against a normalizing flow (NF) using
Neural Spline Flows~\cite{durkan2019neural}
with 8~coupling layers and 2.4M~parameters,
and the Christoffel-function
method~\cite{devonport2021christoffel}
with VC-dimension sample bounds.
All methods share the same trajectory data.

Table~\ref{tab:baseline} compares the three methods.
Per-query inference cost of the DDPM score on a
single NVIDIA H200 is $0.14$\,ms
at batch size $10{,}000$.

\begin{table}[t]
    \centering
    \caption{Baseline comparison on the Duffing oscillator.
        All methods except Christoffel use LTT calibration.}
    \label{tab:baseline}
    \small
    \setlength{\tabcolsep}{4pt}
    \renewcommand{\arraystretch}{1.2}
    \begin{tabular}{@{} l ccc @{}}
        \toprule
        Method & IoU & Precision & FNR(\%) \\
        \midrule
        Christoffel~\cite{devonport2021christoffel}
               & 0.381 & 0.381 & 0.00 \\
        NF
               & 0.689 & 0.726 & 0.08 \\
        DDPM (Ours)
               & \textbf{0.887} & \textbf{0.918} & 0.08 \\
        \bottomrule
    \end{tabular}
\end{table}

The DDPM score, whose sublevel sets approximate
density level sets by
Proposition~\ref{prop:score_density},
achieves the highest IoU~$0.887$ and
precision~$0.918$, while
the observed FNR of $0.08\%$ stays below
$\alpha{=}0.10\%$.
Since $p_0$ is uniform on $\mathcal{X}_0$
and $\mathrm{div}(f) = -c = -0.02$ is constant
for the Duffing oscillator,
Remark~\ref{rem:volume} provides an additional
geometric guarantee:
with $c_0 = 1/\mathrm{vol}(\mathcal{X}_0) = 1/4$,
$\mathrm{vol}(A_k) \leq 4\alpha e^{-0.02t_k}
\leq 4\alpha = 0.004$.
At $t{=}30$ the missed volume is at most
$0.0022$, i.e.\ $0.055\%$ of the initial set area.

The NF uses negative log-likelihood, which also
reflects density in principle.
However, a bijective map from a unimodal prior
cannot faithfully represent the disconnected
branches of the Duffing attractor at late
times~\cite{durkan2019neural,rezende2015variational},
which limits its IoU to~$0.689$.

The Christoffel method's VC-dimension sample
bound~\cite{devonport2021christoffel}
requires $d{=}1$ with our calibration set.
Since this bound is known to be
conservative~\cite{devonport2021christoffel},
we sweep over $d \in \{1,\ldots,20\}$
without enforcing it.
Beyond $d{=}14$, Runge-type oscillatory artifacts
emerge and IoU declines monotonically
to $0.170$ at $d{=}20$.
The sweep selects $d{=}11$ as optimal,
yet the resulting $78$~monomials still yield
a predicted set $2.6{\times}$ the true
reachable set with IoU~$0.381$.

Fig.~\ref{fig:duffing_reach} visualizes the predicted
reachable sets at six time steps.
As the chaotic attractor develops interior voids
and thin filaments ($k{\geq}149$),
the DDPM score resolves these fine structures,
correctly excluding gaps between disjoint branches.
The NF progressively fills in the voids
at large~$k$, and the Christoffel polynomial
over-approximates the reachable set.
\begin{figure*}[t]
    \centering
    \setlength{\tabcolsep}{2pt}
    \renewcommand{\arraystretch}{0}
    \begin{tabular}{@{}c@{\;}c@{\;}c@{\;}c@{\;}c@{\;}c@{\;}c@{}}
                                                                             & \small $k{=}49$ & \small $k{=}99$ & \small $k{=}149$ & \small $k{=}199$ & \small $k{=}249$ & \small $k{=}299$ \\[2pt]
        \rotatebox{90}{\small\;DDPM+LTT}                                     &
        \includegraphics[width=0.13\textwidth]{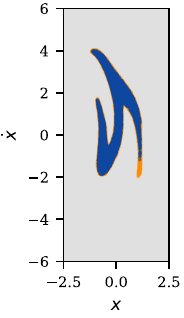}                &
        \includegraphics[width=0.13\textwidth]{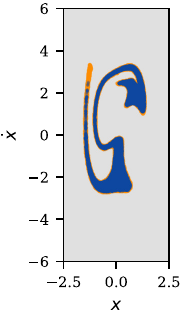}                &
        \includegraphics[width=0.13\textwidth]{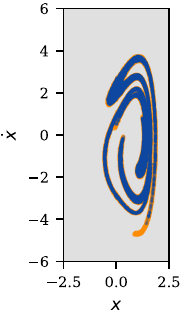}                &
        \includegraphics[width=0.13\textwidth]{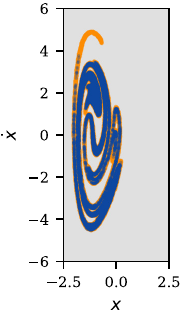}                &
        \includegraphics[width=0.13\textwidth]{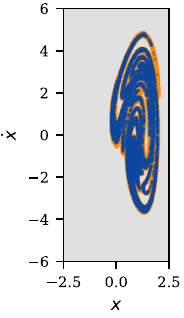}                &
        \includegraphics[width=0.13\textwidth]{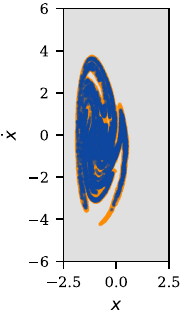}                                                                                                                                 \\[1pt]

        \rotatebox{90}{\small\;NF+LTT}                                       &
        \includegraphics[width=0.13\textwidth]{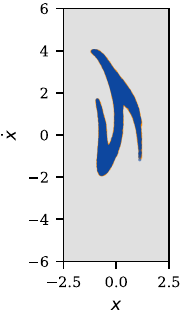}                  &
        \includegraphics[width=0.13\textwidth]{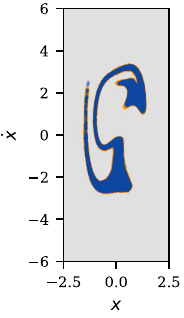}                  &
        \includegraphics[width=0.13\textwidth]{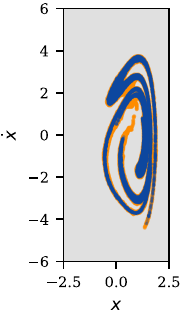}                  &
        \includegraphics[width=0.13\textwidth]{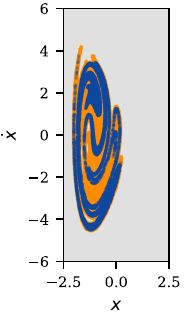}                  &
        \includegraphics[width=0.13\textwidth]{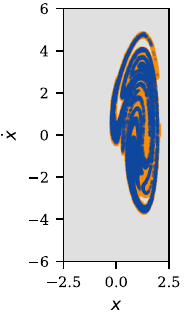}                  &
        \includegraphics[width=0.13\textwidth]{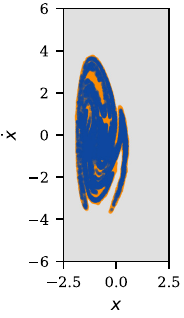}                                                                                                                                   \\[1pt]
        \rotatebox{90}{\small\;Christoffel~\cite{devonport2021christoffel}}                                  &
        \includegraphics[width=0.13\textwidth]{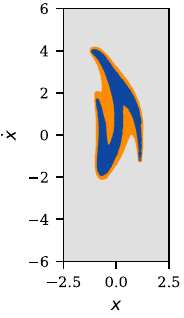} &
        \includegraphics[width=0.13\textwidth]{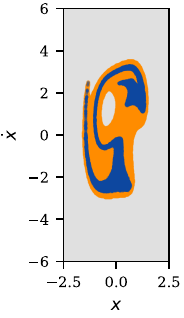} &
        \includegraphics[width=0.13\textwidth]{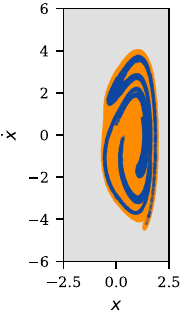} &
        \includegraphics[width=0.13\textwidth]{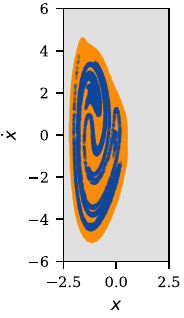} &
        \includegraphics[width=0.13\textwidth]{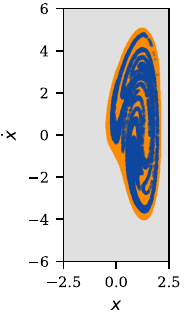} &
        \includegraphics[width=0.13\textwidth]{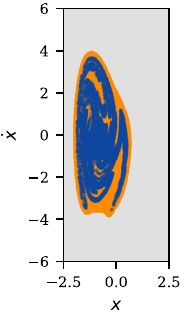}
    \end{tabular}
    \includegraphics[width=0.52\textwidth]{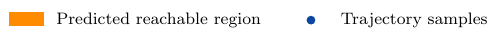}
    \caption{Predicted reachable sets for the Duffing oscillator
        at six time steps~$k$.}
    \label{fig:duffing_reach}
\end{figure*}

The top half of Table~\ref{tab:n_ablation} shows the effect
of the number of training trajectories~$N$ on prediction quality.

\begin{table}[t]
    \centering
    \caption{Ablation studies on the Duffing oscillator.}
    \label{tab:n_ablation}
    \small
    \setlength{\tabcolsep}{4pt}
    \renewcommand{\arraystretch}{1.15}
    \begin{tabular}{@{} l c l r c r @{}}
        \toprule
        Ablation & $N$ & Size & Params & IoU & $T_{\mathrm{query}}$(ms) \\
        \midrule
        \multirow{6}{*}{Data Size} & $10$   & \multirow{6}{*}{$2048{\times}12$} & \multirow{6}{*}{107.5M} & 0.070 & \multirow{6}{*}{0.137} \\
        & $10^2$ & & & 0.133 & \\
        & $10^3$ & & & 0.329 & \\
        & $10^4$ & & & 0.605 & \\
        & $10^5$ & & & 0.803 & \\
        & $10^6$ & & & \textbf{0.887} & \\
        \midrule
        \multirow{5}{*}{Model Capacity} & \multirow{5}{*}{$10^6$} & $128{\times}2$   & 160K   & 0.522 & 0.001 \\
        & & $256{\times}4$   & 847K   & 0.808 & 0.003 \\
        & & $512{\times}6$   & 4.1M   & 0.867 & 0.010 \\
        & & $1024{\times}8$  & 19.1M  & 0.881 & 0.032 \\
        & & $2048{\times}12$ & 107.5M & \textbf{0.887} & 0.137 \\
        \bottomrule
    \end{tabular}
\end{table}
IoU increases monotonically with~$N$, with the
sharpest gain between $N{=}10^3$ and $10^4$.
With only $N{=}10^4$ trajectories the DDPM already
approaches the NF baseline
reported in Table~\ref{tab:baseline}
with IoU${}=0.689$,
while using $100{\times}$ fewer trajectories.
Across all~$N$ and all model sizes, the FNR remains
constant at~$0.08\%$,
because coverage depends solely on
the calibration set, not on model quality.
The bottom half of Table~\ref{tab:n_ablation} shows that a
$512{\times}6$ model with 4.1M parameters already achieves
IoU${}=0.867$, within $0.020$ of
the full $2048{\times}12$ model, while requiring
$26{\times}$ fewer parameters.
Fig.~\ref{fig:n_ablation} visualizes the predicted
reachable set at $k{=}149$.
At $N{=}10^3$ the predicted region already approaches
the Christoffel baseline trained on the full dataset,
and increasing $N$ to $10^6$ progressively refines
boundary details and interior structure.
Even the smallest $128{\times}2$ network resolves
the interior void, and wider architectures sharpen
fine-scale features.
\begin{figure*}[t]
    \centering
    \setlength{\tabcolsep}{1pt}
    \renewcommand{\arraystretch}{0}
    \begin{tabular}{@{}c@{\;}c@{\;}c@{\;}c@{\;}c@{\;}c@{}}
        \small $N{=}10$ & \small $N{=}10^2$ & \small $N{=}10^3$ & \small $N{=}10^4$ & \small $N{=}10^5$ & \small $N{=}10^6$ \\[2pt]
        \includegraphics[width=0.13\textwidth]{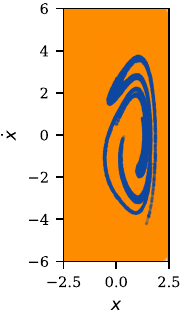}     &
        \includegraphics[width=0.13\textwidth]{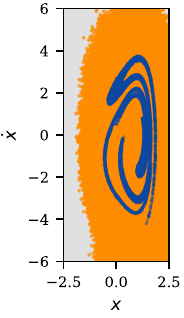}    &
        \includegraphics[width=0.13\textwidth]{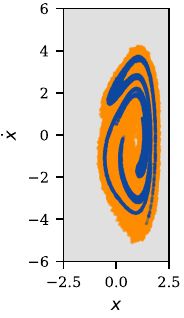}   &
        \includegraphics[width=0.13\textwidth]{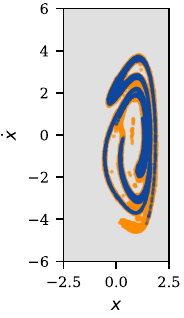}  &
        \includegraphics[width=0.13\textwidth]{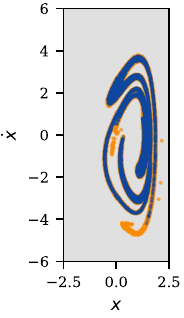} &
        \includegraphics[width=0.13\textwidth]{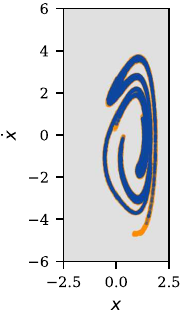} \\[4pt]
    \end{tabular}\\[4pt]
    \begin{tabular}{@{}c@{\;}c@{\;}c@{\;}c@{\;}c@{}}
        \small $128{\times}2$ & \small $256{\times}4$ & \small $512{\times}6$ & \small $1024{\times}8$ & \small $2048{\times}12$ \\[2pt]
        \includegraphics[width=0.13\textwidth]{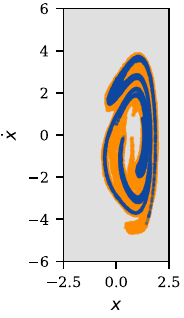}  &
        \includegraphics[width=0.13\textwidth]{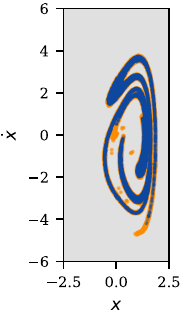}   &
        \includegraphics[width=0.13\textwidth]{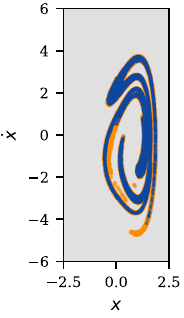}    &
        \includegraphics[width=0.13\textwidth]{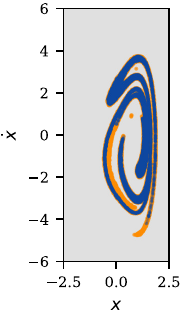}    &
        \includegraphics[width=0.13\textwidth]{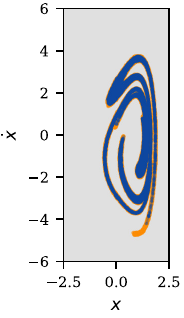}
    \end{tabular}\\[2pt]
    \includegraphics[width=0.52\textwidth]{legend.pdf}
    \caption{Ablation visualizations at $k{=}149$.
        Top row: effect of training data size~$N$
        (model $2048{\times}12$).
        Bottom row: effect of model capacity
        (hidden dim ${\times}$ layers, $N{=}10^6$).}
    \label{fig:n_ablation}
\end{figure*}

To empirically verify the $\delta$-level guarantee,
we repeat the calibration--evaluation procedure
$500$ times with independent random splits
of a $5{\times}10^6$-point pool into
disjoint calibration and test sets.
All $500$ trials satisfy
$\mathrm{FNR} \le \alpha$
with mean~$0.050\%$ and maximum~$0.062\%$,
verifying the PAC guarantee,
which requires a pass rate of at least $1{-}\delta = 80\%$.

\subsection{Planar Quadrotor}
\label{sec:quadrotor}

The planar quadrotor model from
Devonport et al.~\cite{devonport2021christoffel}
has full state
$\bm{x}=(x,h,\theta,\dot{x},\dot{h},\dot{\theta})^\top
\in\R^6$,
where $x$ is the horizontal position,
$h$ is the altitude,
and $\theta$ is the pitch angle.
The dynamics are
\begin{align}
    \ddot{x}      &= u_1 K \sin\theta, \label{eq:quad_x}\\
    \ddot{h}      &= -g + u_1 K \cos\theta, \label{eq:quad_h}\\
    \ddot{\theta}  &= -d_0\theta - d_1\dot{\theta} + n_0 u_2,
    \label{eq:quad_theta}
\end{align}
with gravitational acceleration $g{=}9.81$,
rotor gain $K{=}0.89/1.4$,
and damping parameters $d_0{=}70$, $d_1{=}17$, $n_0{=}55$.
The constant inputs $u_1$ and $u_2$
represent the thrust magnitude and the desired
pitch angle, respectively.
The initial set is a hyperrectangle
$\mathcal{X}_0 \subset \R^6$
following~\cite{devonport2021christoffel}.
The safety specification concerns the
horizontal position and altitude only,
so the reachable set is analyzed in the
reduced state space spanned by $x$ and $h$
at the final time $t_1{=}5.0$.

We generate $N{=}10^6$ trajectories with
inputs drawn uniformly from
$u_1 \in [-1.5 + g/K,\, 1.5 + g/K]$
and $u_2 \in [-\pi/4,\pi/4]$,
and record each trajectory at the terminal time only.
The denoising network is a 6-layer FiLM-MLP
with $T{=}1{,}000$ diffusion steps,
linear $\beta$ schedule, hidden dimension 512,
64-dimensional time embeddings,
and positional encoding with 10 frequency bands,
trained for 500 epochs with AdamW,
lr${=}5{\times}10^{-4}$, batch size $65{,}536$.

The Christoffel-function
baseline~\cite{devonport2021christoffel}
operates on the reduced state $(x,h)$
with polynomial degree $d{=}4$,
which requires $15$ monomials.
At $\alpha{=}0.10\%$ the VC-dimension sample bound
requires $N{\geq}2{,}399{,}222$,
which exceeds the available training data,
so the Christoffel result does not carry
a formal guarantee.

The DDPM method achieves IoU~$0.827$ on the $(x,h)$ projection.
The observed miss rate is $0.091\%$,
below the target level $\alpha{=}0.10\%$.
The Christoffel function achieves IoU~$0.583$.
The degree-4 polynomial sublevel set captures the
overall shape of the reachable set but
cannot resolve fine boundary
details that the DDPM score recovers.
Fig.~\ref{fig:quadrotor_reach} compares the two
predicted sets.

\begin{figure}[t]
    \centering
    \setlength{\tabcolsep}{1pt}
    \begin{tabular}{@{}cc@{}}
        \includegraphics[width=0.49\columnwidth]{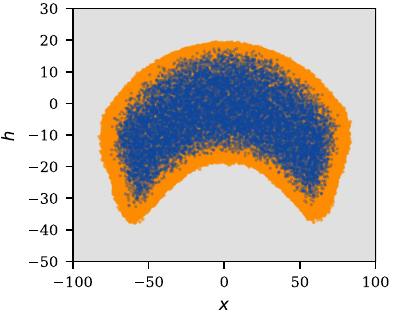} &
        \includegraphics[width=0.49\columnwidth]{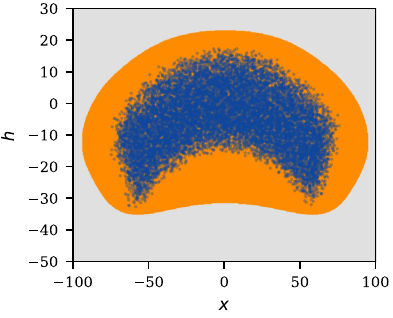} \\
        {\small DDPM+LTT} & {\small Christoffel~\cite{devonport2021christoffel}}
    \end{tabular}
    \caption{Predicted reachable sets for the
        planar quadrotor projected onto
        $(x,h)$ at $t{=}5.0$.}
    \label{fig:quadrotor_reach}
\end{figure}

\subsection{Gray-Scott Reaction-Diffusion}
\label{sec:gray_scott}

The Gray-Scott model~\cite{gray1983autocatalytic}
describes a two-component reaction-diffusion system on a
$64 \times 64$ periodic domain, yielding states
$\bm{x} \in \R^{2 \times 64 \times 64}$:
\begin{align}
    \frac{\partial u}{\partial t}
     & = D_u \nabla^2 u - u v^2 + F(1 - u),
    \label{eq:gs_u}                           \\
    \frac{\partial v}{\partial t}
     & = D_v \nabla^2 v + u v^2 - (F + \kappa)v.
    \label{eq:gs_v}
\end{align}
We set $D_u{=}0.2$, $D_v{=}0.1$, $F{=}0.055$,
$\kappa{=}0.062$, in the spot pattern regime, and simulate
$N{=}10{,}000$ trajectories for $K{=}50$ time steps,
where each step corresponds to 200 PDE solver iterations,
covering the transient to steady-state transition.

The denoising model is a conditional U-Net with
$T{=}1{,}000$ diffusion steps,
linear $\beta$ schedule, base channel width 64
and channel multipliers $(1,2,4)$,
where the physical time enters through cross-attention,
trained for 50 epochs with
AdamW, lr${=}10^{-4}$, batch size 16.
As the only baseline that scales to this dimensionality,
we compare against a variational autoencoder (VAE) with a CNN
encoder-decoder and FiLM time
conditioning~\cite{perez2018film} that uses the
negative ELBO as its nonconformity score,
following the reconstruction probability
framework~\cite{an2015variational}.
Unlike the DDPM, which evaluates the denoising error
across multiple noise scales, the VAE produces a
single-step reconstruction, testing whether the
multi-scale structure of diffusion models provides
a richer score signal.

On a held-out test set of $10^5$ state-time pairs,
the DDPM achieves $\mathrm{FNR}{=}0.008\%$,
below the
target level of~$\alpha{=}0.10\%$.
Over $500$ independent calibration-test splits,
$491$ satisfy $\mathrm{FNR} \leq \alpha$,
for a pass rate of~$98.2\%$,
well above the $1{-}\delta = 80\%$
theoretical floor.
Per-query inference cost on a single NVIDIA H200
at batch size $128$ is $122$ms for the DDPM.

Since grid-based IoU is infeasible
in $n{=}8{,}192$ dimensions,
we evaluate prediction set tightness via
the acceptance rate (AR) under perturbation:
real
training samples are corrupted with additive Gaussian noise
of magnitude $\sigma \cdot \bm{\sigma}_{\mathrm{px}}$,
where $\bm{\sigma}_{\mathrm{px}}$ is the per-pixel
standard deviation of the training data, and the fraction
classified as reachable is recorded as a function
of~$\sigma$.

Fig.~\ref{fig:gs_sensitivity}a shows the result.
The DDPM score exhibits a steeper and earlier
transition than the VAE: AR drops from $100\%$ to
$0\%$ within a narrow band around
$\sigma{\approx}0.01$.
The VAE baseline produces a gradual transition,
AR remains above $99\%$ until $\sigma{\approx}0.04$
and declines to $0\%$ only at $\sigma{\approx}0.2$,
making it roughly $20{\times}$ less sensitive
than the DDPM.

By the SNR decomposition~\eqref{eq:snr_decomp},
the score grows with
$\lVert\bm{x} - \hat{\bm{x}}_0\rVert^2$,
which in high dimensions scales as $\sigma^2 n$
for a Gaussian perturbation of
magnitude~$\sigma$~\cite{graham2023denoising}.
Fig.~\ref{fig:gs_sensitivity}b,c shows the score
distributions for four categories:
clean in-set states, perturbed in-set states at
$\sigma{=}0.01$ and $\sigma{=}0.02$, and
out-of-distribution samples generated with
PDE parameters $F{=}0.025$, $\kappa{=}0.060$, which
produce stripe patterns instead of the training spots.
The DDPM score assigns distinct levels to all
four groups, and the calibrated threshold falls
between $\sigma{=}0.01$ and $\sigma{=}0.02$ states,
closer to $\sigma{=}0.01$, so perturbations at
$\sigma{\geq}0.02$ and different-parameter samples
are reliably rejected.
In contrast, the VAE score distribution shows
heavy overlap among clean, $\sigma{=}0.01$, and
$\sigma{=}0.02$ states, all remaining below the
calibrated threshold.

\begin{figure*}[t]
    \centering
    \begin{subfigure}[t]{0.33\textwidth}
        \centering
        \includegraphics[width=\textwidth]{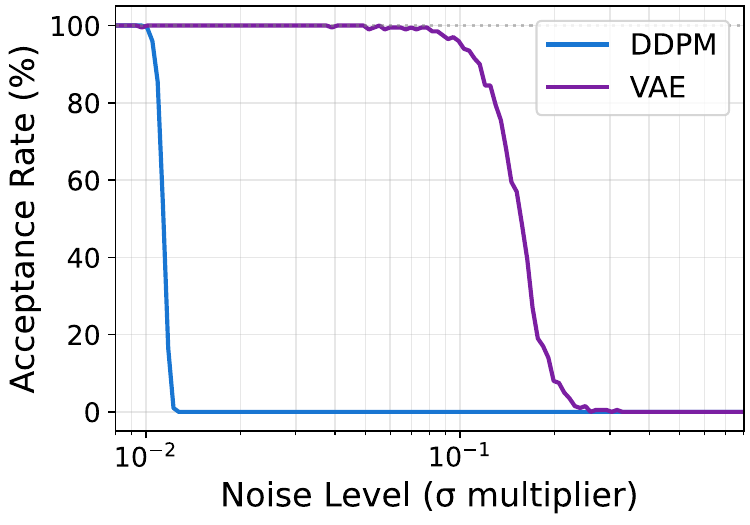}
        \vspace{-6mm}
        \caption{Acceptance rate vs.\ perturbation}
    \end{subfigure}\hfill
    \begin{subfigure}[t]{0.33\textwidth}
        \centering
        \includegraphics[width=\textwidth]{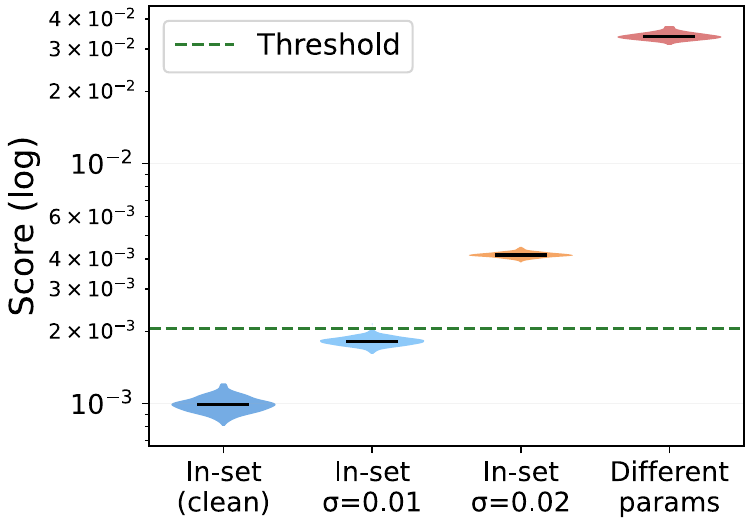}
        \vspace{-6mm}
        \caption{DDPM score distribution}
    \end{subfigure}\hfill
    \begin{subfigure}[t]{0.33\textwidth}
        \centering
        \includegraphics[width=\textwidth]{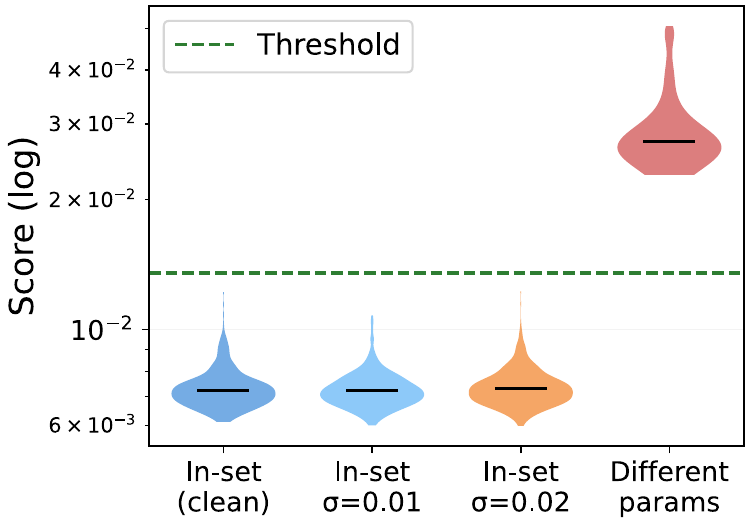}
        \vspace{-6mm}
        \caption{VAE score distribution}
    \end{subfigure}
    \caption{Sensitivity analysis for the Gray-Scott system
        at $t{=}29$.
        (a)~Acceptance rate under Gaussian perturbation.
        (b),(c)~Score distributions for clean in-set,
        perturbed in-set, and different-parameter samples.}
    \label{fig:gs_sensitivity}
\end{figure*}

\section{Conclusion}
\label{sec:conclusion}

We introduced a model-free reachability analysis
framework with PAC coverage guarantees that scales
to state dimensions beyond the reach of classical
grid-based and polynomial methods.
The method operates directly on trajectory data,
calibrating diffusion-based reconstruction error
via LTT to bound the miss rate simultaneously
over the entire time horizon.
Under a mild density lower-bound assumption,
the distributional guarantee upgrades to a
geometric volume bound on the missed portion
of the reachable set,
with the bound tightening
for dissipative systems via Liouville's formula.
Experiments on a forced Duffing oscillator,
a planar quadrotor, and
a Gray-Scott reaction-diffusion system with
8192-dimensional states confirm the coverage
guarantee and the advantage of the
diffusion-based score over baselines.

\bibliographystyle{IEEEtran}
\bibliography{ref}

\end{document}